\documentclass[envcountsame,runningheads]{article}

\usepackage{paralist}
\makeatletter
\newcommand{\xRightarrow}[2][]{\ext@arrow 0359\Rightarrowfill@{#1}{#2}}
\makeatother

\usepackage{varwidth}
\usepackage{enumitem}

\usepackage[a4paper]{geometry}

\usepackage{amsmath,amsthm}
\usepackage{amssymb} 
\usepackage{graphicx}
\usepackage{algorithm}
\usepackage{algorithmicx}
\usepackage[noend]{algpseudocode}
\usepackage{color}
\usepackage{xspace}
\usepackage{pgf}
\newcommand\pfun{\mathrel{\ooalign{\hfil$\mapstochar\mkern5mu$\hfil\cr$\to$\cr}}}
\newcommand{\dom}{\operatorname{\mathbf{dom}}}

\usepackage{tikz}
\usetikzlibrary{arrows,snakes,backgrounds,automata,shapes,positioning,chains}

\usepackage[inference]{semantic}

\usepackage{marginfix} 

\newcommand{\OMIT}[1]{}

\newcounter{thmctr}

\newtheorem{theorem}[thmctr]{Theorem}
\newtheorem{lemma}[thmctr]{Lemma}

\newtheorem{corollary}[thmctr]{Corollary}

\theoremstyle{definition}
\newtheorem{definition}[thmctr]{Definition}
\newtheorem{example}[thmctr]{Example}

\usepackage{cutwin}
 
\usepackage[affil-it]{authblk}

\begin{document}

\title{Causal Linearizability: Compositionality for
  Partially Ordered Executions}

\author{Simon Doherty\thanks{\tt s.doherty@sheffield.ac.uk}}
\author{John Derrick\thanks{\tt j.derrick@sheffield.ac.uk}}
\affil{Department of Computer Science, University of Sheffield, UK}
\author{Brijesh Dongol\thanks{\texttt{Brijesh.Dongol@brunel.ac.uk}}}
\affil{Department of Computer Science, Brunel University London, UK}

\author{Heike Wehrheim\thanks{\texttt{wehrheim@upb.de}}}
\affil{Department of Computer Science, University of Paderborn, Germany}

\newcommand{\linefill}{\cleaders\hbox{$\smash{\mkern-2mu\mathord-\mkern-2mu}$}\hfill\vphantom{\lower1pt\hbox{$\rightarrow$}}}  
\newcommand{\Linefill}{\cleaders\hbox{$\smash{\mkern-2mu\mathord=\mkern-2mu}$}\hfill\vphantom{\hbox{$\Rightarrow$}}}  
\newcommand{\transi}[2]{\mathrel{\lower1pt\hbox{$\mathrel-_{\vphantom{#2}}\mkern-8mu\stackrel{#1}{\linefill_{\vphantom{#2}}}\mkern-11mu\rightarrow_{#2}$}}}
\newcommand{\trans}[1]{\transi{#1}{{}}}
\newcommand{\transo}{\mathord{\trans{~}}}
\newcommand{\ntransi}[2]{\mathrel{\lower1pt\hbox{$\mathrel-_{\vphantom{#2}}\mkern-8mu\stackrel{#1}{\linefill_{\vphantom{#2}}}\mkern-8mu\nrightarrow_{#2}$}}}
\newcommand{\ntrans}[1]{\ntransi{#1}{{}}}
\newcommand{\eseq}[1]{\langle~\rangle}

\newcounter{sarrow}
\newcommand\strans[1]{%
  \mathrel{\raisebox{0.1em}{
    \stepcounter{sarrow}%
    \!\!
    \begin{tikzpicture}
      \node[inner sep=.5ex] (\thesarrow) {$\scriptstyle #1$};
      \path[draw,<-,decorate,line width=0.25mm,
      decoration={zigzag,amplitude=0.7pt,segment length=1.2mm,pre=lineto,pre length=4pt}] 
      (\thesarrow.south east) -- (\thesarrow.south west);
    \end{tikzpicture}%
  }}}

\newcommand\Strans[1]{%
\mathrel{\raisebox{0.1em}{
\!\!\begin{tikzpicture}
  \node[inner sep=0.6ex] (a) {$\scriptstyle #1$};
  \path[line width=0.2mm, draw,implies-,double distance between line
  centers=1.5pt,decorate, 
    decoration={zigzag,amplitude=0.7pt,segment length=1.2mm,pre=lineto,
    pre   length=4pt}] 
    (a.south east) -- (a.south west);
\end{tikzpicture}}%
}}

\newcommand{\calE}{{\cal E}}
\newcommand{\nat}{\mathbb{N}}
\newcommand{\noteq}{\neq}

\newcommand{\lt}{{\bf Less than}}

\newcommand{\ev}{\mathit{ev}}
\newcommand{\Events}{\mathit{Evt}}
\newcommand{\Inv}{\mathit{Inv}}
\newcommand{\Resp}{\mathit{Res}}
\newcommand{\his}{\mathit{his}}
\newcommand{\exec}{\mathit{exec}}
\newcommand{\complete}{\mathit{complete}}
\newcommand{\Var}{\mathit{Loc}}
\newcommand{\Val}{\mathit{Val}}
\newcommand{\CASOp}{\mathit{CAS}} 
\newcommand{\WR}{\mathsf{WR}}
\newcommand{\RA}{\mathsf{RA}}
\newcommand{\R}{\mathsf{R}}
\newcommand{\RX}{\mathsf{RX}}
\newcommand{\W}{\mathsf{W}}
\newcommand{\WX}{\mathsf{WX}}
\newcommand{\C}{\mathsf{C}}
\newcommand{\CRA}{\mathsf{CRA}}
\newcommand{\URA}{\mathsf{U}}
\newcommand{\URAT}{\mathsf{UT}}
\newcommand{\URAF}{\mathsf{UF}}

\newcommand{\HB}{{\sf hb}\xspace} 
\newcommand{\PO}{{\sf po}\xspace}
\newcommand{\MO}{{\sf mo}\xspace}
\newcommand{\SC}{{\sf sc}\xspace}
\newcommand{\RF}{{\sf rf}\xspace}
\newcommand{\SB}{{\sf sb}\xspace}

\newcommand{\ok}{{\tt ok}}
\newcommand{\retabort}{{\tt abort}}
\newcommand{\retempty}{{\tt empty}}

\newcommand{\start}[1]{\texttt{B}_{#1}}
\newcommand{\starto}[1]{(\start{#1},\ok)}
\newcommand{\starta}[1]{(\start{#1},\retabort)}

\newcommand{\commitWriter}[1]{\texttt{CW}_{#1}}
\newcommand{\commitReadOnly}[1]{\texttt{CRO}_{#1}}

\newcommand{\commit}[1]{\texttt{C}_{#1}}
\newcommand{\commito}[1]{(\commit{#1},\ok)}
\newcommand{\commita}[1]{(\commit{#1},\retabort)}

\newcommand{\readcall}[1]{\texttt{R}_{#1}}
\newcommand{\reado}[3]{(\readcall{#1}(#2), #3)}
\newcommand{\reada}[2]{(\readcall{#1}(#2),\retabort)}

\newcommand{\writecall}[1]{\texttt{W}_{#1}}
\newcommand{\writeo}[3]{(\writecall{#1}(#2,#3),\ok)}
\newcommand{\writeos}[3]{\writecall{#1}(#2,#3)}
\newcommand{\writea}[3]{(\writecall{#1}(#2,#3), \retabort)}

\newcommand{\support}{\mathit{support}}

\newcommand{\reads}{\mathit{Reads}}
\newcommand{\writes}{\mathit{Writes}}
\newcommand{\commits}{\mathit{Commits}}

\newcommand{\readas}{\mathit{AReads}}
\newcommand{\writeas}{\mathit{AWrites}}

\newcommand{\idle}{\mathit{idle}}
\newcommand{\live}{\mathit{live}}

\newcommand{\fv}{\mathit{fv}}

\newcommand{\refeq}[1]{(\ref{#1})}
\newcommand{\refalg}[1]{Algorithm~\ref{#1}}

\newcommand{\fr}{{\sf fr}}
\newcommand{\ltsb}{{\sf sb}}
\newcommand{\ltrf}{\mathord{\sf rf}}
\newcommand{\ltfr}{{\sf fr}}
\newcommand{\lthb}{{\sf hb}}
\newcommand{\ltsw}{{\sf sw}}
\newcommand{\ltmox}{{\sf mo}^x}
\newcommand{\ltmo}{{\sf mo}}
\newcommand{\PreExec}{{\it PreExec}}
\newcommand{\Approx}{{\it C11}}
\newcommand{\Seq}{{\it Seq}}

\newcommand{\True}{{\it true}}
\newcommand{\False}{{\it false}}

\newcommand{\justified}{justified\xspace}
\newcommand{\notjustified}{unjustified\xspace}
\newcommand{\Justified}{Justified\xspace}
\newcommand{\Notjustified}{Unjustified\xspace}

\newcommand{\rdval}{{\it rdval}}
\newcommand{\wrval}{{\it wrval}}
\newcommand{\loc}{{\it loc}}

\newcommand{\imp}{\Rightarrow}
\newcommand{\expr}{\mathit{Exp}}

\newcommand*{\HardArrow}{\mathbin{\tikz
    [baseline=-0.15ex,-latex] \draw[->] (0pt,0.5ex) --
    (1.1em,0.5ex);}}%

\newcommand*{\SoftArrow}[1][]{\mathbin{\tikz
    [baseline=-0.15ex,-latex,densely dashed, ->] \draw[->] (0pt,0.5ex) --
    (1.1em,0.5ex);}}%

\newcommand*{\DoubleArrow}{\mathbin{\tikz
    [baseline=-0.25ex,-latex] \draw[->>] (0pt,0.5ex) --
    (1.1em,0.5ex);}}%

\newcommand{\seqarrow}{\DoubleArrow}
\newcommand{\seqtempord}{\DoubleArrow}
\newcommand{\seqcausord}{\prec}
\newcommand{\seqrestr}{\mathop{\downharpoonright}}

\newcommand{\action}[3]{\ensuremath{
\begin{array}[t]{l@{~}l}
\multicolumn{2}{l}{#1}\\
\textsf{Pre:}&#2\\
\textsf{Eff:}&#3
\end{array}
}}

\newcommand{\beginIdx}{{\it beginIdx}}
\newcommand{\readIdx}{{\it readIdx}}
\newcommand{\validatedIdx}{{\it validatedIdx}}

\newcommand{\status}{{\it status}}
\newcommand{\wrSet}{{\it wrSet}}
\newcommand{\rdSet}{{\it rdSet}}
\newcommand{\Used}{{\it Used}}
\newcommand{\TransId}{{\it TransId}}
\newcommand{\earlier}{{\it pastWrites}}
\newcommand{\later}{{\it laterWrites}}
\newcommand{\mems}{{\it mems}}
\newcommand{\store}{{\it store}}

\newcommand{\insertAt}{{\it insertAt}}
\newcommand{\insertAfter}{{\it insertAfter}}

\newcommand{\restrict}{\downharpoonright}

\newcommand{\ES}{\mathsf{ES}}

\newcommand{\bbA}{\mathbb{A}}
\newcommand{\bbB}{\mathbb{B}}
\newcommand{\bbC}{\mathbb{C}}
\newcommand{\bbD}{\mathbb{D}}
\newcommand{\bbE}{\mathbb{E}}
\newcommand{\bbF}{\mathbb{F}}
\newcommand{\bbS}{\mathbb{S}} 
\newcommand{\bbO}{\mathbb{O}} 

\newcommand{\causalord}[1]{\prec_{#1}}
\newcommand{\causalC}{\causalord{\bbC}}
\newcommand{\causalLC}{\causalord{{\cal L}, \bbC}}

\newcommand{\po}{{\sf po}}

\newcommand{\Comm}{Comm}
\newcommand{\PComm}{Comm_{\parallel}}
\newcommand{\whilestep}[1]{\stackrel{#1}{\Longrightarrow}}
\newcommand{\parastep}[1]{\stackrel{#1}{\rightarrow_{\parallel}}}
\newcommand{\whileEnd}{\mathbf{end}}

\newcommand{\traceArrow}[1]{\stackrel{#1}{\longrightarrow}}
\newcommand{\ltsArrow}[1]{\stackrel{#1}{\Longrightarrow}}

\newcommand{\lfState}{{\it State}}
\newcommand{\lfStart}{{\it Init}}
\newcommand{\lfP}{{\it P}}
\newcommand{\lfStep}{{\it step}}
\newcommand{\lfM}{{\it M}}
\newcommand{\lfBody}{{\it body}}
\newcommand{\lfDep}{{\it dep}}
\newcommand{\lfExt}{{\it X}}
\newcommand{\lfEnabled}{{\it Enabled}}

\newcommand{\libf}{\mathcal}
\newcommand{\libField}[2]{{#1}_{\libf{\!#2}}\xspace}
\newcommand{\libState}{\libField{\lfState}}
\newcommand{\libStart}{\libField{\lfStart}}
\newcommand{\libP}{\libField{\lfP}}
\newcommand{\libStep}{\libField{\lfStep}}
\newcommand{\libM}{\libField{\lfM}}
\newcommand{\libBody}{\libField{\lfBody}}
\newcommand{\libDep}{\libField{\lfDep}}
\newcommand{\libExt}{\libField{\lfExt}}
\newcommand{\libEnabled}{\libField{\lfEnabled}}
\newcommand{\libT}[2]{\stackrel{#1}{\libStep{#2}}}

\newcommand{\kwtag}{{\it tag}}
\newcommand{\tid}{{\it tid}}
\newcommand{\act}{{\it act}}
\newcommand{\Op}{A}

\newcommand{\emptymap}{emptymap}
\newcommand{\emptyord}{empord}

\newcommand{\Empty}{\retempty}

\maketitle



\begin{abstract}
In the interleaving model of concurrency, where events are totally ordered,
linearizability is {\em compositional}: the composition of two linearizable
objects is guaranteed to be linearizable. However, linearizability is not
compositional when events are only partially ordered, as in many weak-memory
models that describe multicore memory systems. In this paper,
we present  \emph{causal linearizability}, a correctness condition
for concurrent objects implemented in weak-memory models.
We abstract from the details of specific
memory models by defining our condition using Lamport's
execution structures. We apply our condition to the C11 memory model,
providing a correctness condition for C11 objects. We develop
a proof method for verifying objects implemented in C11 and related models.
Our method is an adaptation of simulation-based methods, but in contrast to
other such methods, it does not require that the implementation totally order
its events. We also show that causal linearizability reduces to linearizability
in the totally ordered case.



\end{abstract}

\newcommand{\eqrng}[2]{(\ref{#1}-\ref{#2})}
\newcommand{\refprop}[1]{Proposition~\ref{#1}}
\newcommand{\reffig}[1]{Fig.~\ref{#1}}
\newcommand{\refthm}[1]{Theorem~\ref{#1}}
\newcommand{\reflem}[1]{Lem\-ma~\ref{#1}}
\newcommand{\refcor}[1]{Corollary~\ref{#1}}
\newcommand{\refsec}[1]{Section~\ref{#1}}
\newcommand{\refex}[1]{Example~\ref{#1}}
\newcommand{\refdef}[1]{Definition~\ref{#1}}
\newcommand{\reflst}[1]{Listing~\ref{#1}}
\newcommand{\refchap}[1]{Chapter~\ref{#1}}
\newcommand{\reftab}[1]{Table~\ref{#1}}

\tikzset{
    mo/.style={dotted,->,>=stealth},
    hb/.style={solid,->,>=stealth},
    rf/.style={dashed,->,>=stealth}
 }

\section{Introduction}

Linearizability \cite{HeWi90,DBLP:books/daglib/0020056} is a
well-studied~\cite{DongolD15} condition that defines correctness of a
concurrent object in terms of a sequential specification. It ensures
that for each history of an implementation, there is a history of the
specification such that (1) each thread makes the same method
invocations in the same order, and (2) the order of non-overlapping
operation calls is preserved. The condition however, critically
depends on the existence of a total order of memory events (e.g., as
guaranteed by sequential consistency (SC)
\cite{DBLP:journals/tc/Lamport79}) to guarantee contextual refinement
\cite{DBLP:journals/tcs/FilipovicORY10} and compositionality
\cite{HeWi90}. Unfortunately most modern systems can only guarantee a
partial order of memory events, e.g., due to the effects of relaxed
memory
\cite{DBLP:journals/computer/AdveG96,AlglaveMT14,DBLP:conf/popl/BattyOSSW11,DBLP:conf/popl/MansonPA05}. It
is known that a naive adaptation of linearizability to the partially
ordered setting of weak memory is problematic from the perspective of
contextual refinement~\cite{DongolJRA18}. In this paper, we propose
a compositional generalisation of linearizability for partially ordered
executions.

\begin{figure}[t] 
\begin{minipage}{0.42\columnwidth}
\centering
$ x= 0, y = 0 \quad $

$\begin{array}{l|l}
  \hfill \textrm{Process 1} \hfill & \hfill \textrm{Process 2} \hfill \\
\hline 
  1: x := 1; \ \ & 1: \mathtt{if} (y = 1)  \quad \\
  2: y := 1; & 2:    \ \ \mathtt{assert} (x = 1); \\
\end{array}$
\caption{Writing to shared variables}
\label{fig:litmus}
\end{minipage}
\ \  
\begin{minipage}{0.5\columnwidth}
\centering
$ \textsc{S} = \langle~\rangle , \textsc{S}' = \langle~\rangle\quad $

$\begin{array}{l|l}
  \hfill \textrm{Process 1} \hfill & \hfill \textrm{Process 2} \hfill \\
\hline 
  1: \textsc{S.Push}(1); \ \ & 1: \mathtt{if } (\textsc{S'.Pop} = 1
                               ) \\
  2: \textsc{S'.Push}(1); & 2:    \ \ \mathtt{assert} (\textsc{S.Pop}
                            = 1 
                            ); \\
\end{array}$
\caption{
  Writing to shared stacks }
\label{fig:litmus-b}
\end{minipage}
\end{figure}

Figs.\ \ref{fig:litmus} and \ref{fig:litmus-b} show two
examples\footnote{Example inspired by H.-J.Boehm talk
  at Dagstuhl, Nov.~2017} of multi-threaded programs on which weak
memory model effects can be observed. \reffig{fig:litmus} shows two
threads writing to and reading from two shared variables $x$ and
$y$. Under SC, the \texttt{assert} in process 2 never fails: if $y$
equals 1, $x$ must also equal 1. However, in weak memory models like
the C11
model~\cite{DBLP:conf/popl/BattyOSSW11,DBLP:conf/popl/LahavGV16}, this
is not true: if the writes to $x$ and $y$ are \emph{relaxed}, process
$2$ may observe the write to $y$, yet also observe the initial
value~$x$ (missing the write to $x$ by process $1$).

Such effects are not surprising to programmers familiar with memory
models~\cite{DBLP:conf/popl/BattyOSSW11,DBLP:conf/popl/LahavGV16}.
However, programmer expectations for linearizable objects, even in a
weak memory model like C11, are different
: if the two stacks $\textsc{S}$ and
$\textsc{S}'$ in \reffig{fig:litmus-b} are linearizable, the
expectation is that the \texttt{assert} will never fail since
linearizable objects are expected to be {\em compositional}
\cite{HeWi90,DBLP:books/daglib/0020056}, i.e., any combination of
linearizable objects must itself be linearizable. However, it is
indeed possible for the two stacks to be linearizable (using the
classical definition), yet for the program to generate an execution in
which the \texttt{assert} fails. The issue here is that
linearizability, when naively applied to a weak memory setting, allows
too many operations to be considered ``overlapping''.

Our key contribution in this paper is the development of a new
compositional notion of correctness, called {\em causal
  linearizability}, which is defined in terms of an {\em execution
  structure}~\cite{DBLP:journals/dc/Lamport86}, taking two different
relations over operations into account: a ``precedence order''
(describing operations that are ordered in time) and a ``communication
relation''. Applied to \reffig{fig:litmus-b}, for a weak memory
execution in which the \texttt{assert} fails, the execution restricted
to stack \textsc{S} would not be causally linearizable in the first
place. Namely, causal linearizability ensures enough \emph{precedence
  order} in an execution to ensure that the method call
$\textsc{S.Push}(1)$ occurs before $\textsc{S.Pop}$, meaning
$\textsc{S.Pop}$ is forced to return $1$.

Execution structures are generic, and can be constructed for any weak
memory execution that includes method invocation/response events. Our
second contribution is one such scheme for mapping executions to
execution structures based on the happens-before relation of the
C11 memory model. Given method calls $m1$ and $m2$, we say $m1$
precedes $m2$ if the response of $m1$ happens before the invocation
$m2$; we say $m1$ communicates with $m2$ if the invocation of $m1$
happens before the response of $m2$.

Our third contribution is a new inductive simulation-style proof
technique for verifying causal linearizability of weak memory
implementations of concurrent objects, where the induction is over
linear extensions of the happens-before relation. This is the first
such proof method for weak memory, and one of the first that enables
full verification, building on existing techniques for linearizability
in SC
\cite{DBLP:journals/tocl/SchellhornDW14,DongolD15,DohertyGLM04}. Our
fourth contribution is the application of this proof technique to
causal linearizability of the Treiber Stack in the C11 memory model.

We present our motivating example, the Treiber Stack in C11 in
\refsec{sec:treiber-stack-c11}; describe the problem of
compositionality and motivate our execution-structure based solution
in \refsec{sec:executionstructures}; and formalise causal
linearizability and prove compositionality in
\refsec{sec:causal-atomicity}. Causal linearizability for C11 is
presented in \refsec{sec:caus-line-c11}, and verification of the stack
described in \refsec{sect:verification}.

\section{Treiber Stack in C11}
\label{sec:treiber-stack-c11}

\begin{algorithm}[t] 
  \caption{Release-Acquire Treiber Stack}
  \label{alg:treiber}
  \begin{algorithmic}[1]
    \Procedure{\textsc{Init}}{}
    \State Top := null; 
    \EndProcedure
    \algstore{tml-ctr}
  \end{algorithmic}
  \smallskip 
  \begin{varwidth}[t]{0.45\columnwidth}
    \begin{algorithmic}[1]
      \algrestore{tml-ctr}
      \Procedure{\textsc{Push}}{v}
      \State n := new(node) ; 
      \State n.val := v ; \label{write-val}
      \Repeat 
      \State top :=$^{\sf A}$ Top ;  \label{push-acquire}
      \State n.nxt := top  ; \label{write-nxt}
      \Until {CAS$^{\sf R}$(\&Top, top, n) } \label{push-release}
      \EndProcedure
      \algstore{tml-ctr}
    \end{algorithmic}

  \end{varwidth}
  \hfill
  \begin{varwidth}[t]{0.55\columnwidth}
    \begin{algorithmic}[1]
      \algrestore{tml-ctr}
      \Procedure{\textsc{Pop}}{} 
      \Repeat 
      \Repeat \label{spin1}
        \State top :=$^{\sf A}$ Top ; \label{pop-acquire} \label{spin2}
      \Until {top $\neq$ null} ; \label{spin3}
      \State ntop := top.nxt ;
      \Until {CAS$^{\sf R}$(\&Top, top, ntop)} \label{pop-release}
      \State \Return top.val ;
      \EndProcedure 
    \end{algorithmic}
  \end{varwidth}
\end{algorithm}

The example we consider (see \refalg{alg:treiber}) is the well-studied
Treiber Stack \cite{Tre86}, executing in a recent version of the C11
\cite{DBLP:conf/pldi/LahavVKHD17} memory model.  In C11, commands may
be annotated, e.g., {\sf R} (for release) and {\sf A} (for acquire),
which introduces extra synchronisation, i.e., additional order over
memory
events~\cite{DBLP:conf/popl/BattyOSSW11,DBLP:conf/popl/LahavGV16}. We
assume racy read and write accesses that are not part of an annotated
command are \emph{unordered} or \emph{relaxed}, i.e., we do not
consider the effects of non-atomic
operations~\cite{DBLP:conf/popl/BattyOSSW11}. Full details of the C11
memory model are deferred until \refsec{sec:caus-line-c11}.

Due to weak memory effects, the events under consideration, including
method invocation and response events are partially ordered. As we
show in \refsec{sec:executionstructures}, it turns out that one cannot
simply reapply the standard notion of linearizability in this weaker
setting; compositionality demands that we use modified form: causal
linearizability that additionally requires ``communication'' across
conflicting operations.

In \refalg{alg:treiber}, all accesses to the shared variable Top are
via an annotated command. Thus, any read of Top (lines
\ref{push-acquire}, \ref{pop-acquire}) reading from a write to Top
(lines~\mbox{\ref{push-release}, \ref{pop-release}}) induces
\emph{happens-before order} from the write to the read. This order, it
turns out, is enough to guarantee invariants that are in turn strong
enough to guarantee\footnote{Note that a successful CAS operation
  comprises both a read and a write access to Top, but we only require
  release synchronisation here. The corresponding acquire
  synchronisation is provided via the earlier read in the same
  operation. This synchronisation is propagated to the CAS by
  \emph{sequenced-before} (aka program order), which, in C11, is
  included in happens-before (see \refsec{sect:verification} for
  details).} causal linearizability of the Stack (see
\refsec{sect:verification}).

Note that we modify the Treiber Stack so that the \textsc{Pop}
operation blocks by spinning instead of returning empty. This is for
good reason - it turns out that the standard Treiber Stack (with a
non-blocking \textsc{Pop} operation) is \emph{not} naturally
compositional if the only available synchronisation is via
release-acquire atomics (see \refsec{sec:synchr-pitf}). 

\usetikzlibrary{positioning}

\section{Compositionality and execution structures}
\label{sec:executionstructures}

This section describes the problems with compositionality for
linearizability of concurrent objects under weak execution
environments (e.g., relaxed memory) and motivates a generic solution
using \emph{execution structures} 
\cite{DBLP:journals/dc/Lamport86}.

\medskip\noindent{\bf Notation.}  First we give some basic
notation. Given a set $X$ and a relation $r \subseteq X \times X$, we
say $r$ is a \emph{partial order} iff it is reflexive, antisymmetric
and transitive, and a \emph{strict order}, iff it is irreflexive,
antisymmetric and transitive. The support of $r$ is denoted
$\support(r) = dom(r) \cup ran(r)$. A partial or strict order $r$ is a
\emph{total order} iff either $(a, b) \in r$ or $(b, a) \in r$ for all
$a, b \in \support(r)$. We typically use notation such as $<$, $\le$,
$\prec$, $\HardArrow$ to denote orders, and write, for example,
$a < b$ instead of $(a, b) \in \mathop<$.  \qed\medskip

The \emph{operations} of an object are defined by a set of {\em
  labels}, $\Sigma$. For concurrent data structures,
$\Sigma = \Inv \times \Resp$, where $\Inv$ and $\Resp$ are sets of
invocations and responses (including their input and return values),
respectively. For example, for a stack $\textsc{S}$ of naturals, the
{\em invocations} are given by
$\{\textsc{Push}(n) \mid n \in \nat\} \cup \{\textsc{Pop}\}$, and the
{\em responses} by $\nat \cup \{\bot, {\tt empty}\}$, and
\begin{align*}
  \Sigma_\textsc{S} = {} & \{(\textsc{Push}(n),\bot), (\textsc{Pop}, n) \mid
                           n\in \mathbb{N}\} \cup
                           \{(\textsc{Pop},\retempty)\}
\end{align*}
The standard notion of linearizability is defined for a concurrent
history, which is a sequence (or total order) of {\em invocation} and
{\em response} events of operations. 

Since operations are concurrent, an invocation of an operation may not
be directly followed by its matching response, and hence, a history
induces a partial order on operations. For linearizability, we focus
on the {\em real-time} partial order (denoted $\HardArrow$), where,
for operations $o$ and $o'$, we say $o \HardArrow o'$ in a history iff
the response of operation $o$ \emph{happens before} the invocation of
operation $o'$ in the history. A concurrent implementation of an
object is linearizable if the real-time partial order ($\HardArrow$)
for \emph{any} history of the object can be extended to a total order
that is \emph{legal} for the object's specification~\cite{HeWi90}. It
turns out that linearizability in this setting is
\emph{compositional}~\cite{HeWi90,DBLP:books/daglib/0020056}: any
history of a family of linearizable objects is itself guaranteed to be
linearizable.

Unfortunately, histories in modern executions contexts (e.g., due to
relaxed memory or distributed computation) are only partially ordered
since processes do not share a single global view of time. It might
seem that this is unproblematic for linearizability and that the
standard definition can be straightforwardly applied to this weaker
setting. However, it turns out that a naive application fails to
satisfy {\em compositionality}.  To see this, consider the following
example.

\begin{example}
  \label{ex:stack1}
  Consider a history $h$, partially ordered by a \emph{happens-before}
  relation, for two stacks {\sc S} and {\sc S'} that are both
  initially empty (denoted by $\bot$). Suppose that in $h$, the
  response of {\sc S'.Push} happens before the invocation of {\sc
    S.Pop}, and the response of {\sc S.Push} happens before the
  invocation of {\sc S.Pop}. History $h$ induces a partial order over
  these operations as shown below:
  
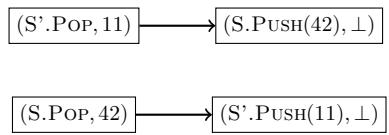
\begin{center}
    \begin{tikzpicture}
        \node [draw,scale=0.8] (A) {$(\textsc{S'.Pop}, 11)$};
        \node [draw, right=of A,scale=0.8] (B) {$(\textsc{S.Push}(42), \bot)$};
        \node [draw, below=0.7cm of A,scale=0.8] (C)
        {$(\textsc{S.Pop}, 42)$};
        \node [draw, right=of C,scale=0.8] (D) {$(\textsc{S'.Push}(11), \bot)$};
        \draw [thick,->] (A) -- (B) ;
        \draw [thick,->] (C) -- (D) ;       
      \end{tikzpicture}
  \end{center}
  If we restrict the execution above to {\sc S} only, we can obtain a
  legal stack behaviour by linearizing $(\textsc{S.Push}(42), \bot)$
  before $(\textsc{S.Pop}, 42)$ without contradicting the real-time
  partial order $\HardArrow$ in the diagram above. Similarly, the
  execution when restricted to $\textsc{S}'$ is linearizable. However,
  the full execution is not linearizable: ordering both pushes before
  both pops contradicts the induced real-time partial order
  ($\HardArrow$ above). \qed
\end{example}

A key contribution of this paper is the development of a correctness
condition, \emph{causal linearizability}, that recovers
compositionality of concurrent objects with partially ordered
histories. Our definition is based on two main insights.

The first insight is that one must augment the real-time partial order
with additional information about the underlying concurrent
execution. In particular, one must introduce information about the
\emph{communication} when linearizing \emph{conflicting}
operations. Two operations conflict if they do not commute according
to the sequential specification, e.g., for a stack data structure,
$\textsc{Push}$ and $\textsc{Pop}$ are conflicting.  Causal linearizability
states that for any conflicting operations, say $o$ and $o'$, that are
linearized in a particular order, say $o\, o'$, there must exist some
communication from $o$ to $o'$. We represent communication by a
relation $\SoftArrow$.

\newpage
\begin{example}
  \label{ex:comm-edges}
  Consider the partial order in \refex{ex:stack1}. For both stacks
  $\textsc{S}$ and $\textsc{S}'$, the {\sc Push} must be linearized
  before the {\sc Pop}, and hence, we must additionally have
  communication edges as follows: \medskip

  \hfill 
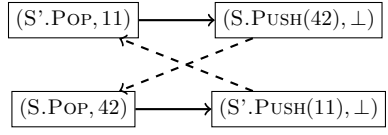
\begin{tikzpicture}
       \node [draw,scale=0.8] (A) {$(\textsc{S'.Pop}, 11)$}; 

       \node [draw, right=of A,scale=0.8] (B) {$(\textsc{S.Push}(42), \bot)$};

       \node [draw, below=0.7cm of A,scale=0.8] (C)
       {$(\textsc{S.Pop}, 42)$};

       \node [draw, right=of C,scale=0.8] (D)
       {$(\textsc{S'.Push}(11), \bot)$};
           
       \draw [thick,->] (A) -- (B) ; 
       \draw [thick,->] (C) -- (D) ;
       \draw [thick,dashed,->] (B) -- (C) ; 
       \draw [thick,dashed,->] (D) -- (A) ;
     \end{tikzpicture} \hfill {} \qed
\end{example}

The second insight is that the operations and the induced real-time
partial order, $\HardArrow$, extended with a communication relation,
$\SoftArrow$, must form an \emph{execution structure}
\cite{DBLP:journals/dc/Lamport86}, defined below.
\begin{definition}[Execution structure]
  Given that $E$ is a finite\footnote{The original presentation allows
    for infinite execution structures, placing a well-foundedness
    condition on $\HardArrow$.} set of events, and
  $\mathord{\HardArrow}, \mathord{\SoftArrow} \subseteq E \times E$
  are relations over $E$, an {\em execution structure} is a tuple
  $(E,\HardArrow,\SoftArrow)$ satisfying the following axioms for
  $e_1, e_2, e_3 \in E$.
  \begin{description}
  \item[A1] The relation $\HardArrow$ is a strict order.
  \item[A2] Whenever $e_1 \HardArrow e_2$, then
    $e_1 \SoftArrow e_2$ and $\neg (e_2 \SoftArrow e_1)$.
  \item[A3]  If $e_1 \HardArrow e_2 \SoftArrow e_3$ or $e_1
    \SoftArrow e_2 \HardArrow e_3$, then $e_1 \SoftArrow
    e_3$.  
  \item[A4] If $e_1 \HardArrow e_2 \SoftArrow e_3 \HardArrow e_4$, then $e_1 \HardArrow e_4$. \qed
  \end{description}
\end{definition}

\begin{example} Consider the execution depicted in
  \refex{ex:comm-edges}.  The requirements of an execution structure,
  in particular axiom {\bf A4} necessitate that we introduce
  additional real-time partial order edges $\mathord{\HardArrow}$ as
  follows.
  
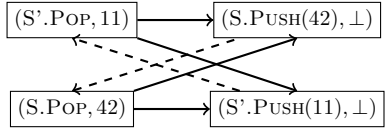
\begin{center}
    \begin{tikzpicture}
       \node [draw,scale=0.8] (A) {$(\textsc{S'.Pop}, 11)$}; 

       \node [draw, right=of A,scale=0.8] (B) {$(\textsc{S.Push}(42), \bot)$};

       \node [draw, below=0.7cm of A,scale=0.8] (C)
       {$(\textsc{S.Pop}, 42)$};

       \node [draw, right=of C,scale=0.8] (D)
       {$(\textsc{S'.Push}(11), \bot)$};

       
       \draw [thick,->] (A) -- (B) ;
       \draw [thick,->] (C) -- (D) ;       
       \draw [thick,dashed,<-] (C.north) -- (B) ;
       \draw [thick,dashed,->] (D) -- (A.south) ;       
       \draw [thick,<-] (B.south) -- (C.north east) ;
       \draw [thick,->] (A.south east)  -- (D.north) ;       
     \end{tikzpicture}
  \end{center}
  For example, the edge
  $(\textsc{S.Pop}, 42) \HardArrow (\textsc{S.Push}(42), \bot)$ is
  induced by the combination of edges
  $(\textsc{S.Pop}, 42) \HardArrow (\textsc{S'.Push}(11), \bot)
  \SoftArrow (\textsc{S'.Pop}, 11) \HardArrow (\textsc{S.Push}(42),\bot) $ together with axiom {\bf A4}. \qed
\end{example}

A consequence of these additional real-time partial order edges is
that {\sc S} (and symmetrically {\sc S'}) is not linearizable since
the edge $(\textsc{S.Pop}, 42) \HardArrow (\textsc{S.Push}(42), \bot)$
must be present even when restricting the structure to {\sc S}
only. Hence compositionality no longer fails.

\newcommand{\legal}{{\it legal}}
\newcommand{\Ws}{{\it Ws}}
\newcommand{\Rs}{{\it Rs}}
\newcommand{\calR}{{\cal R}}

\newcommand{\bbT}{\mathbb{T}}

\newcommand{\History}{\mathit{Hist}} \newcommand{\hs}{\mathit{hs}}

\section{Causal linearizability}
\label{sec:causal-atomicity}
This section provides a formal definition of causal linearizability,
and the compositionality theorem. We define sequential objects in
\refsec{sec:sequ-spec}, then 
define causal linearizability in
\refsec{sec:caus-atom-compl}. 

\subsection{Sequential specifications}
\label{sec:sequ-spec}

Causal linearizability defines correctness of a concurrent object with
respect to a \emph{sequential object} specification.
\begin{definition}[Sequential object]
  \label{def:sequ-object}
  A \emph{sequential object} is a pair $(\Sigma, \legal)$, where
  $\legal \subseteq \Sigma^*$ is a prefix-closed sequence of labels. \qed
\end{definition}
For example, in each legal sequence of a stack, each pop operation
returns the value from the latest push operation that has not yet been
popped, or $\Empty$ if no such operation exists.

For each sequential object, we define a \emph{conflict relation},
$\# \subseteq \Sigma \times \Sigma$, based on the legal behaviours of
the object. Two operations conflict if they do not commute in some
legal history:
\begin{align*}
  o \# o' & = \neg (\forall k_1, k_2 \in \Sigma^*.\  k_1 o\, o' k_2 \in \legal
    \Leftrightarrow k_1 o' o k_2 \in \legal  )
\end{align*}
For a stack, we have, for instance,
$(\textsc{Push}(n), \bot) \# (\textsc{Pop}, n')$ for any $n, n'$, and
for $n \neq n'$,
$(\textsc{Push}(n), \bot) \# (\textsc{Push}(n'), \bot)$.

We now show (in \reflem{lemma:extensions} below) that the order of
conflicting actions in a sequential history captures \emph{all the
  orders in that history that matter}. This is formalized and proved
using order relations derived from a legal sequence. However, since
the same action can occur more than once in a legal sequence, we lift
actions to {\em events} by enhancing each action with a unique tag and
process identifier\footnote{Strictly speaking, the process identifier
  is unimportant for \reflem{lemma:extensions}, but we introduce it
  here to simplify compatibility with the rest of this paper.}. Thus,
given a sequential object $\bbS = (\Sigma, \legal)$ an {\em
  $\bbS$-event} is a triple $(g, p, a)$ where $g$ is an event tag
(taken from a set of tags $G$), $p$ is a process (taken from a set of
processes $P$) and $a$ is a {\em label} in $\Sigma$.  We let $\Events$
be the set of all $\bbS$-events, and for a sequence $k\in \Events^*$,
$\ev(k)$ be the set of events in $k$.

The definitions of legality and conflict as well as sequential
specifications can naturally be lifted to the level of events by
virtue of the action labels. That is, a sequence of events is legal if
the sequence of actions it induces is legal.  In the following, we
therefore use $\legal$ to refer to sequences of actions and sequences
of events interchangeably.  From a sequence $k \in \Events^*$ we
derive two relations on events, a \emph{temporal ordering}
($\mathord{\seqtempord_k}$) and a {\em causal ordering}
($\mathord{\seqcausord_k}$), where:
\begin{align*}
e \seqtempord_k e'& = \exists k_1, k_2, k_3 \in \Events^*.\ k = k_1 e
                    k_2 e' k_3 & \qquad 
e \seqcausord_k e'& = e \seqtempord_k e' \wedge e \# e'
\end{align*}  
{\em Any} sequential history that extends the causal order of a
legal history is itself legal. We formalize this in the following
lemma.

\begin{lemma}[Legal linear extensions]
\label{lemma:extensions}
For a sequential object $(\Sigma, \legal)$, if $k \in \legal$ and
$k'\in\Events^*$, such that $\ev(k) = \ev(k')$ and
$\mathord{\seqcausord_k} \subseteq \mathord{\seqtempord_{k'}}$, then
$k' \in \legal$. 
\end{lemma}
\begin{proof}
  We transform $k$ into $k'$ by reordering events in $k$ to match the
  order $\seqtempord_{k'}$.  We only reorder events that are not
  conflicting and thus, each step of the transformation preserves
  legality.  This is sufficent to prove that $k' \in \legal$.  Let a
  {\em mis-ordered pair} be any pair of events $e, e'$ such that
  $e \seqtempord_{k} e'$ but $e' \seqtempord_{k'} e$. Note that in
  this case, we have $e\not\seqcausord_k e'$, because
  $\mathord{\seqcausord_k} \subseteq \mathord{\seqtempord_{k'}}$.  Let
  $e_1, e_2$ be a mis-ordered pair with minimal distance in $k$
  between the two elements (i.e., so that the number of events in $k$
  between $e_1$ and $e_2$ is minimal).  We will reorder
  non-conflicting events in $k$ so as to eliminate this mis-ordered
  pair, or reduce its size without creating a new mis-ordered
  pair. Once all mis-ordered pairs have been eliminated, we will have
  transformed $k$ into $k'$, while preserving the legality of $k$.

  If $e_1$ and $e_2$ are adjacent in $k$, then because
  $e_1\not\seqcausord_k e_2$, we have $\neg (e_1 \# e_2)$, and thus we
  can reorder them to form a new sequence $\legal$ with fewer
  mis-ordered pairs.

  If $e_1' \# e_2$ then we would have
  $e_1 \seqcausord_k e_1' \seqcausord_k e_2$ and so
  $e_1 \seqcausord_k e_2$, which is a contradiction. The same argument
  shows that there is {\em no} event between $e_1'$ and $e_2$ that
  conflicts with both. So let $k''$ be the sequence derived from $k$
  by reordering $e_1'$ forward just past $e_2$.  Note that because
  there were no conflicts, $k'' \in legal$.  It remains to show that
  $k''$ has no mis-ordered pairs that were not already present in $k$.
  This could only happen if there was some $e_2'$ such that
  $e_1' \seqtempord_S e_2' \seqtempord_S e_2$ and
  $e_1' \seqtempord_{k'} e_2'$.  Because $e_1, e_2$ is the mis-ordered
  pair with {\em minimal} gap in $k$, it must be that
  $e_2' \seqtempord_{k'} e_2$, but then $e_2' \seqtempord_{k'} e_1$
  while $e_1 \seqtempord_{k} e_2'$. Thus, in this case, $e_1, e_2'$
  forms a smaller mis-ordered pair, contrary to
  hypothesis. 
\end{proof}

As we shall see, this lemma is critical in the proof of our
compositionality result, \refthm{theorem:compositionality}.

\subsection{Concurrent executions and causal linearizability}
\label{sec:caus-atom-compl}

We now define causal linearizability. For simplicity, we assume
\emph{complete} concurrent executions, i.e., executions in which every
invoked operation has returned. It is straightforward to extend these
notions to cope with incomplete executions. 

In general, executions of concurrent processes might invoke operations
on several concurrent objects.  To capture this, we define a notion of
{\em object family}, which represents a composition of sequential
objects, indexed by some set $X$.
\begin{definition}[Object family]
  Suppose $X$ is an index set. For each $x \in X$, assume a sequential
  object $\bbS_x = (\Sigma_x, \legal_x)$ such that $\Sigma_x$ is
  disjoint from $\Sigma_y$ for all $y \in X \backslash\{x\}$. We
  define the {\em object family over $X$},
  $\bbS_X = (\Sigma_X, \legal_X)$ by:
  \begin{itemize}
  \item $\Sigma_X = \bigcup_x \Sigma_x$
  \item
    $\legal_X = \{k \in \Sigma_X^* \mid \forall x.~k\seqrestr x \in
    \legal_x\}$, where $k\seqrestr x$ is the sequence $k$ restricted
    to actions of object $x$. Thus the set $\legal_X$ contains exactly
    the interleavings of elements of each of the $\legal_x$. \qed
  \end{itemize}
\end{definition}
\noindent
N.B., the pairwise disjointness requirement on $\Sigma_x$ can be
readily achieved by attaching the object identifier $x$ to each
operation in $\Sigma_x$.

An execution structure $(E, \HardArrow, \SoftArrow)$ is a
\emph{complete $\bbS$-execution structure} iff all events in $E$
are $\bbS$-events.
\begin{definition}[Causal linearizability]
\label{def:caus-atom-complete}
Let $\bbS_X = (\Sigma_X, \legal_X)$ be an object family. A complete
$\bbS_X$-execution structure $(E, \HardArrow, \SoftArrow)$ is {\em
  causally linearizable} if there exists a $k \in \legal_X$ with
$\ev(k) = E$ such that
$\mathord{\HardArrow} \subseteq \mathord{\seqarrow_k}$, and
$\mathord{\SoftArrow}\supseteq \mathord{\seqcausord_{k}}$.
\qed
\end{definition}
Condition $\mathord{\HardArrow} \subseteq \mathord{\seqarrow_k}$
ensures that the real-time (partial) order of operations is consistent
with the chosen $k$, while condition
$\mathord{\SoftArrow}\supseteq \mathord{\seqcausord_{k}}$ captures the
idea that the causal ordering in $k$ (i.e., the ordering between
conflicting actions) requires a communication in the concurrent
execution. Causal linearizability for single objects is a special case of
\refdef{def:caus-atom-complete}, where the family is a singleton set.



To establish compositionality, we must first define an object family's
causal ordering. Note that because an object family's $\legal$ set is
just an interleaving of the underlying object's $\legal$ sets,
operations from distinct objects can always be reordered, and
therefore they never conflict. Thus, we have the following lemma. 
\begin{lemma}
  Suppose $\bbS_X = (\Sigma_X, \legal_X)$ is an object family. For any
  $k \in \legal_X$, we have
  $\mathord{\seqcausord_k} = \bigcup_{x\in X} \seqcausord_{k\seqrestr
      x}$. \qed
\end{lemma}
For an object family $\bbS_X = (\Sigma_X, \legal_X)$ and $x \in X$, we
let $\bbE_x$ be the $\bbS_X$-execution structure $\bbE$ restricted to
$\Sigma_x$.
\begin{theorem}[Compositionality]
  \label{theorem:compositionality}
  Suppose $\bbS_X = (\Sigma_X, \legal_X)$ is an object family over
  $X$, and let $\bbE = (E, \HardArrow, \SoftArrow)$ be a complete
  $\bbS_X$-execution structure. Then, $\bbE_x$ is causally linearizable
  w.r.t. $\bbS_x$ for all $x \in X$ iff $\bbE$ is causally linearizable
  w.r.t. $\bbS_X$.
\end{theorem}
\begin{proof}
  The implication from right to left is straightforward.  For the
  other direction, for each $x \in X$, let $k_x$ be the legal
  sequential execution witnessing causal linearizability of
  $\bbE_x$. Let $\rightsquigarrow$ be the irreflexive transitive
  relation defined by
  \[ 
    \mathord{\rightsquigarrow} = (\HardArrow \cup \bigcup_{x\in X}
    \seqcausord_{k_x})^+
  \]
  We show that $\rightsquigarrow$ is acyclic, and is therefore a strict
  partial order. Because $\rightsquigarrow$ is a partial order, there is
  some total order $\seqarrow_m\, \supseteq\, \rightsquigarrow$, where
  $m \in \legal_X$. This total order defines a sequence of labels of
  $\bbS_X$.  We prove that this sequence witnesses the causal
  linearizability of $\bbE$.  By definition, we have
  $\mathord{\seqcausord_{k_x}} \subseteq \mathord{\rightsquigarrow}
  \subseteq \mathord{\seqarrow_m}$ for all $x$, and so by Lemma
  \ref{lemma:extensions}, we
  have 
  $\seqarrow_{m\seqrestr x} \in \legal_x$ (where
  $\seqarrow_{m\seqrestr x}$ is the restriction of $\seqarrow_m$ to the
  events of object $x$). Furthermore,
  \begin{enumerate}
  \item $\mathord{\HardArrow} \subseteq \mathord{\seqarrow_m}$ follows
    from
    $\mathord{\HardArrow} \subseteq \mathord{\rightsquigarrow}
    \subseteq \mathord{\seqarrow_m}$,
  \item $\SoftArrow\ \supseteq \mathord{\seqcausord_{k_x}}$ follows
    from causal linearizability of $\bbE_x$, and hence
    $\SoftArrow\ \supseteq \bigcup_{x \in X} \seqcausord_{k_x}$, as
    required.
\end{enumerate}
Thus, $\bbE$ is causally linearizable, as required.

We show that $\rightsquigarrow$ is acyclic by contradiction. Suppose
$\rightsquigarrow$ contains a cycle. Pick
$P = e_1 \rightsquigarrow e_2 \rightsquigarrow \ldots \rightsquigarrow
e_n$ to be the minimal cycle. Since $\HardArrow$ is acyclic, and each
$\seqcausord_x$ is acyclic, the cycle $P$ must contain accesses to
least two different objects. Without loss of generality, assume $e_1$
and $e_2$ access different objects, i.e., $e_1\in E_y$, $e_2\in E_z$
for some $y \neq z$. Since each $\seqcausord_x$ only orders elements
of $E_x$, we must have $e_1 \HardArrow e_2$.  Observe that $P$ must be
of length greater than two, i.e., it cannot be of the form
$e_1 \rightsquigarrow e_2 \rightsquigarrow e_1$ since we would then
have $e_1 \HardArrow e_2 \HardArrow e_1$, which contradicts the
assumption that $\HardArrow$ is a partial order.
 
 Hence $P$ must contain a third (distinct) element $e_3$. Note that $e_2 \not\HardArrow e_3$,
 because otherwise we could shorten the cycle $P$, using the
 transitivity of $\HardArrow$. Thus $e_2, e_3$ are from the same object $z$
 and $e_2 \seqcausord_{k_z} e_3$. By the causal linearizability
 of $E_z$, we must have $e_2 \SoftArrow e_3$. Let $e$ be the element
 of $P$ following $e_4$ (so possibly $e = e_1$). Note that $e_3 \not\seqcausord_{k_z} e$,
 because otherwise we could shorten the cycle $P$, using the
 transitivity of $\seqcausord_{k_z}$. Thus, $e_3 \HardArrow e$, so we have
 \[
   e_1 \HardArrow e_2 \SoftArrow e_3 \HardArrow e
 \]
 By the execution structure axiom {\bf A4}, we have
 $e_1 \HardArrow e$, and hence there exists a cycle
 $e_1 \rightsquigarrow e \rightsquigarrow \dots \rightsquigarrow e_1$
 contradicting minimality of $P$. 
\end{proof}

\subsection{Relationship with classical linearizability}
\label{sec:libraries}

In this section, we show that classical linearizablity, which is
defined for totally ordered \emph{histories} of invocations (events of
type $\Inv$) and responses (events of type $\Resp$), degenerates to
causal linearizability. As in the previous section, for simplicity, we
assume the histories under consideration are complete; extensions to
cope with incomplete histories are straightforward.

First, we describe a method, inspired by the execution structure
constructions given by Lamport \cite{DBLP:journals/dc/Lamport86}, for
constructing execution structures for any \emph{well-formed partially
  ordered} history.  We let
$\History = (\Inv \cup \Resp) \times (\Inv \cup \Resp)$ denote the
type of all histories. A history is \emph{well-formed} if it is a
partial order and the history restricted to each process is a total
order of invocations followed by their \emph{matching} response. The
set of all matching pairs of invocations and responses in a history
$h$ is given by
$mp(h)$. 
A history is \emph{sequential} iff it is totally ordered and each
invocation is immediately followed by its matching response.  Note
that a history could be totally ordered, but not sequential (as is the
case for the concurrent histories considered under
SC~\cite{DBLP:books/daglib/0020056,HeWi90}).

\begin{definition}
  \label{def:abs-map}
  Let $h \subseteq \History$ be a well-formed (partially ordered)
  history. We say $exec(h) = (E, \HardArrow, \SoftArrow)$ is the
  execution structure corresponding to $h$ if
  \begin{align*}
    E & \ \ = \ \ mp(h) \\
    \HardArrow & \ \ =\ \ \{((i_1, r_1), (i_2, r_2)) \in E \times E
                 \mid (r_1, i_2) \in h \} \\
    \SoftArrow & \ \ =\ \ \{((i_1, r_1), (i_2, r_2)) \in E \times E
                 \mid (i_1, r_2) \in h\} 
  \end{align*}
\end{definition}

We now work towards the standard definition of linearizability.
Recall that a sequential object (see \refdef{def:sequ-object}) is
defined in terms of sequences of labels of type $\Sigma^*$, where
$\Sigma = \Inv \times \Resp$, whereas sequential histories are of type
$\History$. 
Thus, we define a function $\gamma : \History \to \Sigma^*$ such that
for each pair $(i, r), (i', r') \in mp(\hs)$ of sequential history
$\hs$ we have $(r, i') \in \hs$ iff
$(i,r) \seqarrow_{\gamma(\hs)} (i',r')$. Thus, the order of operations
in $hs$ and $\gamma(hs)$ are identical.

A complete history $h$ is {\em linearizable} w.r.t.\ a (family of)
sequential object(s) $\bbS = (\Sigma, \legal)$ iff there exists a
sequential history $\hs$ such that $\gamma(\hs) \in \legal$, for each
process $p$, $h \seqrestr p = \hs \seqrestr p$ and
$\mathord{h} \subseteq \mathord{\hs}$ \cite{HeWi90}.

\begin{theorem} 
  \label{th:lin-inst-atomic} 
  Suppose $h$ is a totally ordered complete history and $\bbS$ a
  (family of) sequential object(s). Then $h$ is linearizable w.r.t
  $\bbS$ iff $exec(h)$ is causally linearizable w.r.t. $\bbS$.  \qed
\end{theorem}



\newcommand{\oproj}{\it oproj}
\newcommand{\TopLoc}{\it \&Top}
\newcommand{\alloc}{\it alloc}

\newcommand{\arlx}{{\sf X}}
\newcommand{\arel}{{\sf R}}
\newcommand{\aacq}{{\sf A}}

\renewcommand{\sb}{{\it sb}} 
\newcommand{\rf}{{\it rf}}
\newcommand{\mo}{{\it mo}} 
\newcommand{\sw}{{\it sw}}
\renewcommand{\fr}{{\it fr}}
\newcommand{\hb}{{\it hb}}

\newcommand{\Mod}{{\it Mod}}
\newcommand{\Qry}{{\it Qry}}
\newcommand{\Loc}{{\it Loc}}
\newcommand{\Alloc}{{\it A}}

\newcommand{\ilabel}[1]{{\bf #1}}

\section{Causal linearizability of C11 implementations}
\label{sec:caus-line-c11}


We now introduce the C11 memory model, where we adapt the programming-language oriented presentation of
C11~\cite{DBLP:conf/popl/LahavGV16,DBLP:conf/vmcai/DokoV16}, but we
ignore various features of C11 not needed for our discussion,
including non-atomic operations and fences. \medskip

\noindent{\bf The C11 memory-model.}
Let $L$ be a set of {\em locations} (ranged over by $x, y$), let $V$
be a set of values (ranged over by $u, v$).  Our model employs a set
of {\em memory events}, which can be partitioned into read events,
$R$, write events, $W$, and update events, $U$. Moreover, let
$\Mod = W \cup U$ be the set of events that modify a location, and
$\Qry = R \cup U$ be the set of events that query a location. For any
memory event $e$, let $loc(e)$ be the event's location, and let
$ann(e)$ be the event's annotation. Let
$\Loc(x) = \{e \mid loc(e) = x\}$.  For any query event let $rval(e)$
be the value read; and for any modification event let $wval(e)$ be the
value written.  An event may carry a {\em synchronisation annotation},
which may either be a release, $\arel$, or an acquire, $\aacq$,
annotation.


A {\em C11 execution} (not to be confused with an {\em execution
  structure}) is a tuple $\bbD = (D, \sb, \rf, \mo)$ where $D$ is a set of
events, and $\sb, \rf, \mo \subseteq D \times D$ define the
\emph{sequence-before}, {\em reads-from} and {\em modification order}
relations, respectively. We say a C11 execution is {\em valid} when it
satisfies:
\begin{enumerate}
\item [\textbf{(V1)}] $\sb$ is a strict order, such that, for each
  process $p$, the projection of $sb$ onto $p$ is a total order;
  
\item [\textbf{(V2)}] for all $(w, r) \in \rf$, $loc(w) = loc(r)$ and
  $wval(w) = rval(r)$; 
\item [\textbf{(V3)}] for all $r \in D \cap Qry$, there
  exists some $w \in D \cap \Mod$ such that $(w, r) \in \rf$;
  
\item [\textbf{(V4)}] for all $(w, w') \in \mo$, $loc(w) = loc(w')$; and
  
\item [\textbf{(V5)}] for all $w, w' \in W$ such that $loc(w) = loc(w')$,
  $(w, w') \in \mo$ or $(w', w) \in \mo$.
\end{enumerate}


Other relations can be derived from these basic relations. For
example, assuming $D_\arel$ and $D_\aacq$ denote the sets of events
with release and acquire annotations, respectively, the {\em
  synchronises-with} relation,
$\sw = \rf \cap (D_\arel \times D_\aacq)$, creates interthread
ordering guarantees based on synchronisation annotations.  The {\em
  from-read} relation, $\fr = (\rf^{-1} ; \mo) \setminus
Id$, 
relates each query to the events in modification order {\em after} the
modification that it read from. Our final derived relation is the {\em
  happens before} relation $\hb = (\sb \cup \sw)^+$, which formalises
causality.  We say that a C11 execution is {\em consistent} if
\begin{itemize}
\item [\textbf{(C1)}] $\hb$ is acyclic, and 
\item [\textbf{(C2)}] $\hb ; (\mo \cup \rf \cup \fr)$ is irreflexive.
\end{itemize}
\medskip

\noindent{\bf Method invocations and responses.} So far, the events
apearing in our memory model are standard.  Our goal is to model
algorithms such as the Treiber stack. Thus, we add {\em method events}
to the standard model, namely, \emph{invocations}, $\Inv$, and
\emph{responses}, $\Resp$. Unlike weak memory at the processor
architecture level, where program order may not be preserved
\cite{DongolJRA18}, program order in C11 is consistent with
happens-before order, and hence, these can be introduced here in a
straightforward manner. The only additional requirement is that
validity also requires \textbf{(V6)} $sb$ for each process projected
restricted $\Inv \cup \Resp$ must be alternating sequence of
invocations and matching responses, starting with an
invocation. \medskip

\noindent{\bf Dynamic memory.}
To describe the behaviour of algorithms, such as the Treiber Stack, we
must define reads and writes to higher-level structures. To this end,
we develop a simple theory of {\em references} to objects, the {\em
  fields} of those objects and \emph{memory allocations} for the
object. We let $F$ be the set of all fields and $\Alloc$ be the set of all
memory allocation events, which  
is an event of the form $A(l)$ for a location $l$.  We let
$. : L \times F \to L$ be the function that returns a location for a
given location, field pair. We use infix notation $x.f$ for $.(x,f)$,
where $x \in L$ and $f \in F$. We then introduce three additional
constraints: \textbf{(A1)} for every $a, a' \in E \cap \Alloc$, if
$loc(a) = loc(a')$ then $a = a'$; and \textbf{(A2)} if $l.f = l'.f'$
then $l = l'$ and $f = f'$. \textbf{(A3)} for all locations $l$ and
fields $f$ there are no allocations of the form $A(l.f)$. \medskip




\noindent\textbf{From C11 executions to execution structures.}
A C11 execution with method invocations and responses naturally gives
rise to an execution structure. First, for a C11 execution $\bbF$, let
the history of $\bbD$, denoted $hist(\bbD)$ be the happens-before
relation for $\bbD$ restricted to the invocation and response
events. By \textbf{(V6)}, $hist(\bbD)$ is a well-formed history. Thus,
we can apply the construction defined in \refsec{sec:libraries}
to build an execution structure $exec(hist(\bbD))$.
\begin{definition}
  We say that a C11 execution $\bbD$ is \emph{causally linearizable}
  w.r.t a sequential object if $exec(hist(\bbD))$ is.
  \qed
\end{definition}
We can now state a compositionality result for a C11 execution $\bbD$
of an object family $X$. The property follows from
\refthm{theorem:compositionality} and the fact that for any object
$x \in X$, $exec(hist(\bbD_x)) = exec(hist(\bbD))_x$, where $\bbD_x$
is $\bbD$ restricted to events of object $x$. Note that $\bbD_x$
contains all events of $x$, i.e., all invocations, responses and
low-level memory operations of $x$. 
\begin{corollary}[Compositionality for C11 executions]
  Suppose that $\bbS_X = (\Sigma_X, \legal_X)$ is an object family
  over $X$, and let $\bbD$ be an execution. Then, $\bbD_x$ is causally
  atomic w.r.t. $\bbS_x$ for all $x \in X$ iff $\bbD$ is causally
  atomic w.r.t. $\bbS_X$.  \qed
\end{corollary}
Finally, note that because the $sb$ relation is included in $hb$,
$hist(\bbD)$ includes program order on the invocations and responses
of $\bbD$.

\newcommand{\stackof}{{\it stackOf}}
\newcommand{\init}{{\it init}}
\newcommand{\rec}{{\it rec}}

\section{Verification}
\label{sect:verification}

We now describe an operational method for proving that a given
C11 execution is causally linearizable w.r.t a given sequential
object. Accordingly, we give a state-based, operational model of a
sequential object that generates legal sequences of labels
(\refdef{def:sequ-object}), then present a simulation-based proof rule
for causal linearizability (\refsec{sec:hb-simulation}). Then, we
illustrate our technique on the Treiber Stack
(\refsec{sec:case-study}).



\subsection{A simulation relation over happens-before}
\label{sec:hb-simulation}

An {\em operational sequential object} is a tuple
$(\Gamma, \init, \tau)$ where: $\Gamma$ is a set of {\em states};
$\init \in \Gamma$ is
the {\em initial state} 
and $\tau : \Gamma \times H\times \Inv \pfun \Gamma \times H$, where
$H = (\Inv \times \Resp)^*$, is a partial {\em update function} that
applies an invocation to a state and a history, returning the
resulting state and updated history. We require that for
$s \in \Gamma$, $h \in H$ and $i \in \Inv$, there exists some
$r \in \Resp$, such that
$\tau(s, h, i) = (\_\, , h \cdot \langle (i, r)\rangle)$, where we use
$\cdot$ for sequence concatenation. This response $r$ is the object's
response to the invocation $i$.

\begin{example}[Operational sequential stack]
\label{ex:op-stack}
  A stack containing natural numbers can be represented as an operational sequential object
  in the following way. Let $\Gamma = \nat^*$,
  $\init = \langle~\rangle$ and define the update function as follows
  \[
    \tau(s, h, \textsc{Push}(n)) = (n \cdot s, h \cdot \langle
    (\textsc{Push}(n), \bot)\rangle) \quad\:\: \tau(n \cdot s, h,
    \textsc{Pop}) = (s, h \cdot \langle (\textsc{Pop}, n) \rangle )
  \]
  for $n \in \nat$ and $s \in \Gamma$. Note that assuming $i$ is a
  stack invocation (as per \refsec{sec:executionstructures}),
  $\tau(s, i)$ is defined iff $s \neq \langle~\rangle$ or
  $i \neq \textsc{Pop}$.
  \qed
\end{example}
Given an operational sequential object $\bbS = (\Gamma, \init, \tau)$,
it is easy to construct a corresponding sequential object (in the
sense of Definition \ref{def:sequ-object}).  Let
$\Sigma_{\bbS} = \Inv \times \Resp$ and let $\legal_{\bbS}$ be the set
of histories returned by $\tau$. Thus $(\Sigma_{\bbS}, \legal_{\bbS})$
is a sequential object, and our method verifies causal linearizability
w.r.t that
object. 

For the remainder of this section, fix a C11 execution
$\bbD = (D, \sb, \rf, \mo)$, and an operational sequential object
$\bbS = (\Gamma, \init, \tau)$.  We describe a method for proving that
$\bbD$ is causally linearizable w.r.t $\bbS$.  Our proof method is an
induction on the length of some linear extension of $\bbD$'s $hb$
order. The proof proceeds by remembering the set $Z \subseteq D$ of
events that have already been considered by the induction, i.e., $Z$
defines the {\em current stage} of the induction. The set $Z$ is
assumed to be downclosed with respect to $\hb$, i.e., if $z \in Z$ and
$(z', z) \in \hb$, then $z' \in Z$.  At each stage of the induction,
we add an arbitrary $e \notin Z$ to $Z$, where $e$'s $\hb$
predecessors are already in $Z$ (i.e., the set $Z \cup \{e\}$ is also
downclosed w.r.t. $\hb$).

Correctness of each inductive step is formalised by a {\em simulation
  relation}, $\rho$, relating the events in the current state, $Z$, to
a state of the operational sequential object. Each inductive step of
the implementation must match a ``move'' of the sequential object,
i.e., be a stutter step, or a state update as given by the update
function of the sequential object. Moreover, assuming that $\rho$
holds for $Z$ (before each inductive step), $\rho$ must hold after the
step (i.e., for $Z \cup \{e\}$).

Following the existing verification literature \cite{DongolD15}, we
refer to events corresponding to non-stuttering steps as
\emph{linearization points}: the points where the high-level operation
appears to take effect. The verifier must define a function
$lp : D \cap Inv \to D$ to determine the memory event that linearizes
the given invocation, and this function must satisfy certain
constraints with respect to the simulation relation $\rho$, as
described in Definition \ref{def:hb-sim}, below.

For each low-level operation, we must also determine the invocation and
response to which it belongs. Thus we also define a function
$\mu : D \to D \cap \Inv$ that maps each event in $D$ to the
invocation responsible for producing $e$, and a function and
$\nu : D \to D \cap \Resp$ that that maps $e$ to the response produced
by $e$'s invocation. More formally, $\mu(e)$ is the latest invocation
in $\sb$-order prior to $e$, and $\nu(e)$ is the earliest response in
$\sb$-order after $e$.



Thus, we obtain the following definition.
\begin{definition}[$\hb$-simulation]
\label{def:hb-sim}
Suppose
  $\bbD = (D, \sb, \rf, \mo)$ is an execution and
  $\bbS = (\Gamma, \init, \tau)$ an operational sequential object.  An
  {\em hb-simulation} is a relation
  $\rho \subseteq 2^D \times (\Gamma \times H)$ such that:
  \begin{enumerate}
  \item \label{hbsim:init} $\rho(\emptyset, (\init, \langle~\rangle))$, and \hfill
    (initialisation)
  \item \label{hbsim:step} for all $Z \subseteq D$, and events $e \in D \setminus Z$ such
    that $Z \cup \{e\}$ is down-closed w.r.t $\bbD$'s happens-before
    order, if $\rho(Z, (s, h))$ then
    \begin{enumerate}
    \item \label{hbsim:stutter} if $e \neq lp(\mu(e))$ then $\rho(Z \cup \{e\}, (s,
      h))$, \hfill (stutter step)
    \item \label{hbsim:lin} if $e = lp(\mu(e))$ \hfill (linearization step)
      
      provided $i = \mu(e)$, $h' = h\cdot \langle (i, r)\rangle$, and
      $\tau(s, h, i) = (s', h')$, then
      \begin{enumerate}
      \item \label{hbsim:lin-pres} $\rho(Z \cup \{e\}, (s', h'))$, and
      \item \label{hbsim:lin-resp} $\nu(e) = r$, and
      \item \label{hbsim:lin-caus} for all operations $(i', r')$ in $h$, if
        $(i',r') \seqcausord_{h'} (i, r)$
        then $(lp(i'), e) \in hb$.
      \end{enumerate}
    \end{enumerate}
  \end{enumerate}
\end{definition}
The initialisation is straightforward, while the two inductive steps
consider a new $e$ for inclusion in $Z$ following $\hb$ order. If $e$
is a stutter step, we only have to prove that $\rho$ is preserved by
adding $e$ to $Z$. If $e$ is a linearization step (that is, if
$e = lp(\mu(e))$), then there are three obligations: prove that $\rho$
is preserved (\ref{hbsim:lin-pres}); prove that the response of the
high-level operation matches that returned by the sequential object
(\ref{hbsim:lin-resp}); and prove that whenever some operation that
has already been linearized is causally prior to the newly linearized
operation, then that operation's linearization point is $hb$-prior to
the new event $e$ (\ref{hbsim:lin-caus}).





\begin{theorem}[Soundness of hb-simulation]
\label{theor:hb-sim}
If $\rho$ is an hb-simulation for a C11 execution $\bbD$, then $\bbD$
is causally linearizable.
\end{theorem}
\begin{proof}The proof below uses a formulation of an operational
  sequential object where $\tau$ that does not maintain a history.

Fix the operational sequential object $\bbS = (\Gamma, \init, \tau)$.
Fix the execution $\bbD$, and let $\leq_E$ be any linear extension of $\bbD$'s
$\hb$ relation. Assume that $lp$ is the linearization function and
$\rho$ is the simulation relation.

We perform an induction on the indexes of $\leq_E$. Let $e_{n}$ be the
nth event in $\leq_E$ order, so we are indexing from $0$. Let $Z_n$ be the
set of events strictly below the nth index. Thus,
\[
Z_n = \{e_m \mid m < n\}
\]
Note that $Z_0 = \emptyset$ and $Z_{n+1} = Z_n \cup \{e_n\}$.  We
define a function $rep : \{n \mid n < |\leq_E|\} \to \Gamma$
recursively as follows:
\begin{align}
rep(0) &= \init &\\
rep(n+1) &= rep(n) &\text{when $lp(\mu(e_{n})) \neq e_{n}$}\\
rep(n+1) &= \pi_1(\tau(rep(n), \mu(e_{n})) &\text{when $lp(\mu(e_{n})) = e_{n}$}
\end{align}
By induction, we have $\rho(Z_n, rep(n))$ for all $n < |\leq_E|$.
\begin{itemize}
\item Because $\rho(Z_0, rep(0)) = \rho(\emptyset, \init)$, Proposition \ref{hbsim:init}
ensures that $\rho(Z_0, rep(0))$.
\item Assume $\rho(Z_n, rep(n))$ and $lp(\mu(e_{n})) \neq e_{n}$. Then
$\rho(Z_{n+1}, rep(n+1)) = \rho(Z_n \cup \{e_n\}, rep(n))$, and thus Property
\ref{hbsim:stutter} ensures that $\rho(Z_{n+1}, rep(n+1))$.
\item Assume $\rho(Z_n, rep(n))$ and $lp(\mu(e_{n})) = e_{n}$. 

Then
$\rho(Z_{n+1}, rep(n+1)) = \rho(Z_n \cup \{e_n\}, \pi_1(\tau(rep(n), \mu(e_{n})))$, and thus Property
\ref{hbsim:lin-pres} ensures that $\rho(Z_{n+1}, rep(n+1))$.
\end{itemize}

We turn now to defining $k$, the legal sequence we need to witness
causal linearizability of $\bbD$.
\begin{align}
k_0 &= \langle\rangle &\\
k_{n+1} &= k_n &\text{when $lp(\mu(e_{n})) \neq e_{n}$}\\
rep(n+1) &= k_n \cdot \langle(\mu(e), \nu(e)) \rangle&\text{when $lp(\mu(e_{n})) = e_{n}$}
\end{align}
It is easy to see that this is a legal history, and that $(k_n, rep(n))$ is a move.

We need to show that $\leq_H \subseteq \seqtempord_k$. Consider a response $r$ and invocation
$i$ such that $(r, i) \in hb$. Ley $i'$ be the invocation of $r$, and let $r'$ be the
response of $i$..
Because $\mu(lp(i)) = i$, we have $\{(lp(i'), r), (i, lp(i))\} \in sb \subseteq hb$,
and thus $(lp(i'), lp(i)) \in hb$ and so $lp(i')$ appears at an earlier point in
$\leq_E$ than $lp(i)$, and therefore $(i', r) \seqtempord_k (i, r')$, as required.

Finally, we must show that $\seqcausord_k \subseteq \SoftArrow_{\bbD}$. This is a
simple induction on the length of $k$, with the hypothesis that, for all
operations $(i, r), (i', r')$ in $k$, if $(i, r) \seqcausord_k (i', r')$ then
$(lp(i), lp(i') \in hb$. At each step we apply Property \ref{hbsim:lin-caus}.
Thus, for each existing operation $(i, r)$ and new operation $(i', r')$,
we have $(i, r) \seqcausord_k (i', r') \implies (lp(i), lp(i') \in hb$
immediately. On the other hand, $(i', r') \seqcausord_k (i, r)$ is impossible,
because $(i', r') \seqtempord_k (i, r)$ is false.

This completes oour proof.
\end{proof}




\subsection{Case-study: the Treiber Stack}
\label{sec:case-study}

We now outline an $\hb$-simulation relation $\rho$ for the Treiber
stack. We fix some arbitrary C11 execution $\bbD = (D, \sb, \rf, \mo)$
that contains an instance of the Treiber stack. That is, the
invocations in $\bbD$ are the stack invocations, and the responses are
the stack responses (as given in Section
\ref{sec:executionstructures}).  Furthermore, the low-level memory
operations between these invocations and responses are generated by
executions of the operations of the Treiber stack
(\refalg{alg:treiber}).

The main component of our simulation relation guarantees
correctness of the \emph{data representation}, i.e., the sequence of values formed by
following next pointers starting with $\&Top$ forms an appropriate stack, and we focus
on this aspect of the relation. As is typical with verifications of shared-memory algorithms, there
are various other properties that would need to be considered in a full proof.

In a sequentially consistent setting, the data
representation can easily be obtained from the state (which maps
locations to values). However, for C11 executions calculating the
data representation requires a bit more work. In what follows, we define
various functions that depend on a set $Z$ of events, representing the current
stage of the induction.

We define the {\em latest write} in $Z$ to a location $x$ as
$latest_Z(x) = max(\mo \seqrestr (Z \cap \Loc(x)))$ and the
\emph{current value} of a location $x$ in some set $Z$ as
$cval_Z(x) = wval(latest_Z(x))$, which is the value written by the
last write to $x$ in modification order.
It is now straightforward to construct the sequence of values
corresponding to a location as
$\stackof_Z(x) = v \cdot \stackof_Z(y)$, where $v = cval_Z(x.val)$ and
$y = cval_Z(x.nxt)$.

Now, assuming that $(s, h)$ is a state of the operational sequential stack, our
simulation requires:
\begin{align}
\label{prop:data-rep}
\stackof_Z(cval_Z(\&Top)) = s
\end{align}
Further, we require that all modifications of $\&Top$ are totally
ordered by $hb$:
\begin{align}
\label{prop:mod-total-ord}
  \forall m, m' \in Z \cap Mod(\&Top).\ (m, m') \in hb \vee (m', m) \in hb 
\end{align}
to ensure that any new read considered by the induction sees the most
recent version of $\& Top$. 

The linearization function $lp$ for the Treiber stack is completely standard:
each operation is linearized at the unique update operation generated by
the unique successful CAS at line 9 (for pushes) or line 16 (for pops).

In what follows, we illustrate how to verify the proof obligations
given in Definition \ref{def:hb-sim}, for the case where the new event
$e$ is a linearization point.  Let $e$ be an update operation that is
generated by the CAS at line 9 of the push operation in 
\refalg{alg:treiber}. The first step is to prove that every
modification of $\&Top$ in $Z$ is happens-before the update event
$e$. Formally,
\begin{align}
\label{prop:after-all}
  \forall m \in Z \cap \Mod \cap \Loc(\&Top).\ (m, e) \in hb
\end{align}
Proving this formally is somewhat involved, but the essential reason
is as follows.  Note that there is an acquiring read $r$ to $\&Top$
executed at line 7 of $e$'s operation and $sb$-prior to $e$.  $r$
reads from some releasing update $u$.  Thus, by Property
\ref{prop:mod-total-ord}, and the fact the $hb$ contains $sb$, $e$ is
happens after $u$, and all prior updates. If there were some update
$u'$ of $\&Top$ such that $(u', e) \notin hb$, then
$(u', u) \notin hb$ so by Property \ref{prop:mod-total-ord},
$(u, u') \in hb$. But it can be shown in this case that the CAS that
generated $e$ could not have succeeded, because $u'$ constitutes an
update intervening between $r$ and $e$. Therefore, there can be no
such $u'$.


Property \ref{prop:after-all}
makes it straightforward to verify that Condition \ref{hbsim:lin-caus} of Definition \ref{def:hb-sim}
is satisfed. To see this, note that every linearization point of every operation
is a modification of $\&Top$. Thus, if $(i', r')$ is some operation such that
$lp(i') \in Z$ (so that this operation has already been linearized) then $(lp(i'), e) \in hb$.

Using Property \ref{prop:after-all} it is easy to see that both
Property \ref{prop:data-rep} and Property~\ref{prop:mod-total-ord} are
preserved. We show by contradiction that $latest_{Z'}(\&Top) =
e$. Otherwise, we have $(e, latest_{Z'}(\&Top)) \in mo$.  Therefore
$(latest_{Z'}(\&Top), e) \notin hb$, but $latest_{Z'}(\&Top)$ is a
modification operation, so this contradicts Property
\ref{prop:after-all}.

It follows from $latest_{Z'}(\&Top) = e$ that
$\stackof(cval_{Z'}) = \stackof(wval(e))$. Given this, it is
straightforward to show that Property \ref{prop:data-rep} is
preserved. This step of the proof relies on certain simple properties
of push operations.  Specifically, we need to show that the current
value of the $val$ field of the node being added to the stack
(formally, $cval_Z((wval(e)).nxt)$) is the value passed to the push
operation; and that the current value of the $nxt$ field (formally,
$cval_Z((wval(e)).nxt)$) is the current value of $\&Top$ when the
successful CAS occurs. These properties can be proved using the model
of dynamic memory given in Section \ref{sec:caus-line-c11}.






\section{A synchronisation pitfall}
\label{sec:synchr-pitf}

We now describe an important observation regarding failure of
compositionality of read-only operations caused by weak memory
effects. The issue can be explained using our abstract notion of an
execution structure, however, a solution to the problem is not
naturally available in C11 with only release-acquire annotations.

Consider the Treiber Stack in \refalg{alg:treiber} that returns empty
instead of spinning; namely where the inner loop (lines
\ref{spin1}-\ref{spin3}) is replaced by code block ``top :=$^{\sf A}$
Top ; {\bf if} top = null {\bf then} {\bf return} {\tt empty}''.  Such an
implementation could produce executions such as the one in
\reffig{fig:isempty-1} which, like the examples in
\refsec{sec:executionstructures}, is not compositional. Recovering
compositionality requires one to introduce additional communication
edges as shown in \reffig{fig:isempty-2}. In the C11 memory model,
these correspond to ``from-read'' anti-dependencies from a read to a
write overwriting the value read. However, release-acquire
synchronisation is not adequate for promoting from-read order in the
memory to happens-before.

    


\begin{figure}[!t]
  \noindent
  \begin{minipage}[t]{0.45\columnwidth}
    \centering
    \scalebox{1}{
      \begin{tikzpicture}
        \node [draw,scale=0.8] (A) {$(\textsc{S.Push}(1), \bot)$};
        \node [draw, right=of A,scale=0.8] (B) {$(\textsc{S'.Pop}, {\tt empty})$};
        \node [draw, below=of A,scale=0.8] (C) {$(\textsc{S'.Push}(2),
          \bot)$};
        \node [draw, right=of C,scale=0.8] (D) {$(\textsc{S.Pop},
          {\tt empty})$};
        
        \draw [thick,->] (A) -- (B) ;
        \draw [thick,->] (C) -- (D) ;       
      \end{tikzpicture}}
      \caption{Read-only operations without communication (not
        compositional)}
      \label{fig:isempty-1}
  \end{minipage}
  \hfill 
  \begin{minipage}[t]{0.47\columnwidth}
    \centering
    \scalebox{1}{
      \begin{tikzpicture}
        \node [draw,scale=0.8] (A) {$(\textsc{S.Push}(1), \bot)$};
        \node [draw, right=of A,scale=0.8] (B)
        {$(\textsc{S'.Pop}, {\tt empty})$}; \node [draw, below=of
        A,scale=0.8] (C) {$(\textsc{S'.Push}, \bot)$}; \node [draw,
        right=of C,scale=0.8] (D)
        {$(\textsc{S.Pop}, {\tt empty})$};
        
        \draw [thick,->] (A) -- (B) ;
        \draw [thick,->] (C) -- (D) ;       
        \draw [thick,dashed,->] (B) -- (C.north) ;
        \draw [thick,<-] (B.south) -- (C.north east) ;
        \draw [thick,dashed,->] (D) -- (A) ;
        \draw [thick,->] (A.south east) -- (D.north) ;       
      \end{tikzpicture}}
      \caption{Read-only operations with communication
        (compositional)}
      \label{fig:isempty-2} 
  \end{minipage}
\end{figure}


One fix would be to disallow read-only operations, e.g., by
introducing a release-acquire CAS operation on a special variable that
always succeeds at the start of each operation. However, such a fix is
somewhat unnatural. Another would be to use C11's {\sf SC}
annotations, which can induce synchronisation across from-read
edges. However, the precise meaning of these annotations is still a
topic of active
research 
\cite{DBLP:conf/pldi/LahavVKHD17,DBLP:conf/popl/BattyDW16}.




\section{Conclusion and related work}

We have presented {\em causal linearizability}, a new correctness
condition for objects implemented in weak-memory models,
that generalises linearizability and addresses the important problem of compositionality.
Our condition is not tied to a particular memory model,
but can be readily applied to memory models, such as C11,
that feature a happens-before relation. We have presented
a proof method for verifying causal linearizability.
We emphasise that our proof method can be applied directly to a
standard axiomatic memory model. Unlike other recent proposals
\cite{DBLP:conf/esop/DokoV17,DBLP:conf/ecoop/KaiserDDLV17}, we model C11's relaxed
accesses without needing to prohibit their problematic dependency cycles
(so called ``load-buffering'' cycles).


Although causal linearizability has been presented as a condition for
concurrent objects, we believe it is straightforward to extend this
condition to cover, for example, transactional memory. We
intend to develop our approach into a framework in which
the behaviour of programs that mix transactional memory, concurrent objects
and primitive weak-memory operations can be precisely described in a compositional
fashion.


Causal linearizability is closely related to \emph{causal
  $\hb$-linearizability} defined in \cite{DongolJRA18}, which is a
causal relaxation of linearizability that uses specifications
strengthened with a happens-before relation. The compositionality
result there requires that either a specification is commuting or that
a client is unobstructive (does not introduce too much
synchronisation). Our result is more general as we place no such
restriction on the object or the client. Others
\cite{DBLP:conf/sefm/DohertyD16} define a correctness condition, also
called causal linearizabilty, that is only compositional when the
client satisfies certain constraints; in contrast, we achieve full
decoupling.  Furthermore, that condition is only defined when the
underlying memory model is given operationally, rather than
axiomatically like C11. Early attempts, targetting TSO architectures,
used totally ordered histories but allowed the response of an
operation to be moved to a corresponding ``flush'' event
\cite{DBLP:conf/wdag/GotsmanMY12,DBLP:conf/esop/BurckhardtGMY12,DBLP:conf/hvc/TravkinW14,DBLP:conf/ictac/DongolDS14}.
Others have considered the effects of linearizability in the context
of a client abstraction. This includes a correctness condition for C11
that is strictly stronger than linearizability under SC
\cite{DBLP:conf/popl/BattyDG13}. Although we have applied causal
linearizability to C11, causal linearizability itself is more general
as it can be applied to any weak memory model with a happens-before
relation. Causal consistency \cite{Ahamad1995} is a related condition,
aimed at shared-memory and data-stores, which has no notion of
real-time order and is not compositional.





\bibliographystyle{plain}

\bibliography{references}

\newpage
\appendix

\section{Potentially incomplete executions}
\label{sec:incompl-exec}

A complication with concurrent executions is that they may contain
\emph{incomplete operations} (that have been invoked, but have not yet
returned). Since the effect of an incomplete operation may be globally
visible, they cannot simply be ignored. This phenomenon has been well
studied and arises in the definitions of linearizability \cite{HeWi90}
and opacity \cite{DBLP:conf/ppopp/GuerraouiK08}. This section
describes how we cope with incomplete operations in the context of
causal linearizability.

We define the \emph{completable extension} of a sequential object
$\bbS = (\Sigma, \legal)$ to be a triple $\bbT = (\bbS, I, C)$, where
$I$ is a set of \emph{allowable incomplete actions} and
$C : I \to 2^\Sigma$ is a \emph{completion function} that maps each
$I$ to a set of possible completions for $I$.
\begin{example}
  If $\bbS = (\Sigma, \legal)$ is a concurrent object the set of
  allowable incomplete operations and completion function is defined
  by:
  \begin{align*}
    I & = \{i \in Inv \mid \exists k_1, k_2 \in \Sigma^*.\ \exists r.\ k_1 (i, r)
        k_2 \in legal\}
    & \quad
    C & = \lambda i.\  \{(i, r) \in \Sigma \}
  \end{align*}
  For the Treiber stack in \refalg{alg:treiber}, we have
  $I_\textsc{S} = \{\textsc{Push}(n) \mid n\in \mathbb{N}\} \cup
  \{\textsc{Pop}
  \}$, and
  $C_\textsc{S} (\textsc{Push}(n)) = \{\bot\}$ and
  $C_\textsc{S} (\textsc{Pop}) = \nat \cup \{\Empty\}$. 
\end{example}

A \emph{completable extension} of an object family $\bbS_X$ is a
triple $\bbT_X = (\bbS_X, I_X, C_X)$, where $I_X = \bigcup_x I_x$ and
$C_X(a) = C_x(a)$, with $x$ being the unique element of $X$ such that
$a \in \Sigma_x$. 

We say that $\mathbb{E} = (E, \HardArrow, \SoftArrow)$ is a
\emph{$\bbT$-execution structure} iff $E \subseteq \Sigma \cup I$ such
that $\dom(\HardArrow) \cap I = \emptyset$, i.e., no element of $E$
may depend (in real-time order) on an element in $I$. Note that there
may be $\SoftArrow$ edges both in and out of elements in $E \cap I$
and $\HardArrow$ edges into $E \cap I$.  A $\bbT$-execution structure
is causally atomic if we can replace all incomplete events by complete
events in a way that is allowed by the corresponding sequential
object. This process is analagous to the {\em extension} of incomplete
histories to complete histories, as allowed by linearizability and
opacity in the classical (i.e., sequentially consistent) setting.

\begin{definition}[Causal linearizability]
  \label{def:caus-atom}
  Let $\bbS = (\Sigma, \legal)$ be a family of sequential objects and
  $\bbT$ its completable extension. A $\bbT$-execution structure
  $\mathbb{E} = (E, \HardArrow, \SoftArrow)$ is {\em causally atomic}
  w.r.t. $\bbS$ iff there exists a causally atomic $\bbS$-execution
  structure $\mathbb{E'} = (E', \HardArrow', \SoftArrow')$ and an
  (order-preserving) isomorphism
  $\varphi : \mathbb{E} \to \mathbb{E}'$ such that:
  \begin{itemize}
  \item for each event $e = (g, p, i) \in E$, where $i \in I$, we have
    $\varphi(e) = (g, p, a)$ for some $a \in C(i)$, and
  \item for each event $e = (g, p, a) \in E$, where $a \in \Sigma$, we
    have $\varphi(e) = (g, p, a)$.
  \end{itemize}
\end{definition}


\refthm{theorem:compositionality} extends directly to the case of
incomplete histories. The fact that each individual object history is
causally atomic implies that we can assume the existence of a valid
extension for each incomplete event, which is itself causally atomic.
Thus, we can apply these per-object extensions to the object-family
execution, and show, using the proof of Theorem
\ref{theorem:compositionality}, that the resulting complete execution
structure is causally atomic.

\end{document}